\documentclass[11pt, fleqn]{article}

\usepackage[letterpaper, margin=1in]{geometry}
\usepackage{graphicx}
\usepackage{hyperref} 
\usepackage{amsmath}
\usepackage{amssymb}
\usepackage{amsfonts}
\usepackage{amsthm}
\usepackage{fullpage}
\usepackage{enumitem}
\usepackage{algorithm}
\usepackage[noend]{algpseudocode}
\let\oldReturn\Return
\renewcommand{\Return}{\State\oldReturn}
\usepackage{xcolor}

\newtheorem{theorem}{Theorem}[section]
\newtheorem{lemma}[theorem]{Lemma}
\newtheorem{proposition}[theorem]{Proposition}
\newtheorem{corollary}[theorem]{Corollary}

\newtheorem{observation}[theorem]{Observation}

\newtheorem{remark}[theorem]{Remark}
\newtheorem{example}[theorem]{Example}

\newcommand{\bx}{\mathbf{x}} 
\newcommand{\by}{\mathbf{y}}
\newcommand{\bp}{\mathbf{p}}

\title{A Polynomial-Time Algorithm for Fair and Efficient Allocation with a Fixed Number of Agents}
\author{Ryoga Mahara\thanks{Department of Mathematical Informatics, University of Tokyo.
E-mail: mahara@mist.i.u-tokyo.ac.jp}
}
\date{}

\begin{document}

\maketitle
\begin{abstract}
We study the problem of fairly and efficiently allocating indivisible goods among agents with additive valuation functions. Envy-freeness up to one good (EF1) is a well-studied fairness notion for indivisible goods, while Pareto optimality (PO) and its stronger variant, fractional Pareto optimality (fPO), are widely recognized efficiency criteria. 
Although each property is straightforward to achieve individually, simultaneously ensuring both fairness and efficiency is challenging. Caragiannis et al.~\cite{caragiannis2019unreasonable} established the surprising result that maximizing Nash social welfare yields an allocation that is both EF1 and PO; however, since maximizing Nash social welfare is NP-hard, this approach does not provide an efficient algorithm. To overcome this barrier, Barman, Krishnamurthy, and Vaish~\cite{barman2018finding} designed a pseudo-polynomial time algorithm to compute an EF1 and PO allocation, and showed the existence of EF1 and fPO allocations.
Nevertheless, the latter existence proof relies on a non-constructive convergence argument and does not directly yield an efficient algorithm for finding EF1 and fPO allocations. 
Whether a polynomial-time algorithm exists for finding an EF1 and PO (or fPO) allocation remains an important open problem.

In this paper, we propose a polynomial-time algorithm to compute an allocation that achieves both EF1 and fPO under additive valuation functions when the number of agents is fixed. Our primary idea is to avoid processing the entire instance at once; instead, we sequentially add agents to the instance and construct an allocation that satisfies EF1 and fPO at each step.
\end{abstract}

\newpage

\section{Introduction}\label{sec:int}
The fair division problem has been a central research topic across various fields, including mathematics, economics and computer science, since it was formally introduced by Steinhaus~\cite{steinhaus}. This problem aims to allocate resources among agents in a {\it fair} and {\it efficient} manner.
It has various real-world applications, including rent division~\cite{edward1999rental}, course allocation~\cite{budish2017course, othman2010finding}, pilot-to-plane assignment for airlines, allocation of tasks to workers, articles to reviewers, and fair recommender systems.
Early studies primarily focused on {\it divisible} goods, resulting in extensive literature on fair division in mathematics and economics~\cite{brams1996fair, robertson1998fair, moulin2004fair, brandt2016handbook}.
A widely accepted standard of fairness is {\it envy-freeness} (EF)~\cite{foley1966resource}, which requires that each agent prefers their own bundle to that of any other agent. 
On the other hand, Pareto optimality (PO) is a fundamental efficiency criterion: an allocation is {\it Pareto optimal} if no allocation exists that makes an agent better off without making any other agent worse off.
Pareto optimality is independent of fairness, and is widely used to evaluate whether allocations are efficient and free of waste.
Varian~\cite{varian1974equity} notably showed that, for divisible goods, there always exists an envy-free and Pareto optimal allocation.
Furthermore, it can be computed in polynomial time under additive valuation functions~\cite{eisenberg1959consensus, devanur2008market, orlin2010improved, vegh2012strongly}.

In contrast, for {\it indivisible} goods, where each good must be allocated to a single agent, envy-free allocations are not guaranteed to exist. For instance, even in the simple case of distributing a single good between two agents, no envy-free allocation is possible. 
This limitation implies that classical fairness concepts and algorithms are often inapplicable to indivisible goods, motivating research into new fairness concepts and algorithms specifically suited to discrete fair division problems. For an overview of recent advances in this area, we refer the reader to recent surveys~\cite{amanatidis2023fair, walsh2021fair}.

To address the fair division problem with indivisible goods, several relaxed fairness notions have been proposed. One well-studied fairness notion is {\it Envy-Freeness up to one good} (EF1), introduced by Budish~\cite{budish2011combinatorial}.
EF1 requires that each agent prefers their own bundle to that of any other agent after removing at most one good from the latter's bundle. Under monotone valuation functions, it has been shown that an EF1 allocation always exists and can be computed in polynomial time~\cite{lipton2004approximately}.

Achieving EF1 or PO individually is straightforward; however, whether EF1 and PO can be attained simultaneously is a key question. Caragiannis et al.~\cite{caragiannis2019unreasonable} established the surprising result that, under {\it additive} \footnote{{\it Additivity} implies that each agent’s valuation for a set of goods equals the sum of their valuations for each individual good within that set.} valuations, maximizing Nash social welfare~\cite{Nash50, kaneko1979nash}, defined as the geometric mean of the agents’ valuations, results in an allocation that is both EF1 and PO. However, since maximizing Nash social welfare is NP-hard~\cite{nguyen2014computational} and even APX-hard~\cite{lee2017apx}, this approach does not directly yield an efficient algorithm.

To overcome this barrier, Barman, Krishnamurthy, and Vaish~\cite{barman2018finding} proposed a pseudo-polynomial time algorithm that computes an EF1 and PO allocation and proved that an EF1 and {\it fractionally Pareto optimal} (fPO) allocation always exists.
An allocation is {\it fractionally Pareto optimal} (fPO) if no fractional allocation exists that makes an agent better off without making any other agent worse off.
Nevertheless, their existence proof relies on a non-constructive convergence argument and does not directly yield an algorithm for computing such an allocation.

Clearly, fPO is a stronger efficiency criterion than PO.
In addtion, fPO allocations offer another advantage over PO allocations: Given an allocation, fPO can be verified efficiently~\cite{Sandomirskiy22}, whereas verifying whether an allocation is PO is coNP-complete~\cite{de2009complexity}. 
This efficient verification is particularly advantageous when a centralized authority conducts the allocation since all participants can verify that the allocation is fPO (and therefore also PO). However, this efficient verification is not feasible for PO allocations.
Whether a polynomial-time algorithm exists for finding an EF1 and PO (or fPO) allocation remains an important open problem.

\subsection{Related Work}

\paragraph{Fair and efficient allocation for indivisible goods}
Caragiannis et al.~\cite{caragiannis2019unreasonable} established that, under additive valuations, maximizing Nash social welfare~\cite{Nash50, kaneko1979nash} yields an allocation that is both EF1 and PO. 

An approach commonly used to find allocations that satisfy both fairness and PO uses the connection between fair division and market equilibrium in Fisher markets~\cite{budish2011combinatorial}.

Barman, Krishnamurthy, and Vaish~\cite{barman2018finding} developed a pseudo-polynomial time algorithm to compute an allocation that is both EF1 and PO. They further showed the existence of an EF1 and fPO allocation.

Barman and Krishnamurthy~\cite{barman2019proximity} developed a strongly polynomial-time algorithm that computes a PROP1 and fPO allocation, where PROP1 (Proportionality up to one good) is a fairness notion that is weaker than EF1 under additive valuation functions.

Garg and Murhekar~\cite{garg2024computing} proposed a pseudo-polynomial time algorithm for computing EF1 and fPO allocations. However, due to significant issues in their proof, we were unable to verify its correctness (see Appendix~\ref{ap:2} for details). They also reported that when the number of agents is fixed, there exists a polynomial-time algorithm to compute an EF1 and PO allocation.

Garg and Murhekar~\cite{garg2023computing} showed that under additive, bi-valued valuation functions, an EFX (envy-free up to any good) and fPO allocation  exists and can be computed in polynomial time.
They also established that EFX and PO are incompatible in instances with three distinct values. Here, EFX, which is a stronger fairness notion than EF1, requires that each agent prefers their own bundle to that of any other agent after removing any single good from the latter’s bundle.

Freeman et al.~\cite{freeman2019} showed that if all values are positive, an EQX+PO allocation always exists; however, they also showed that when values may be zero, an EQ1+PO allocation does not exist.
Here, EQ1 (equitability  up to one good) and EQX (equitability up to any good) refer to equitability criteria.

\paragraph{Fair and efficient allocation for indivisible chores}

The fair division problem for chores (items with negative value) is also an important research topic. In this context, concepts such as EF1 and PO can be defined analogously to their counterparts in the case of goods. However, unlike for indivisible goods, the existence of an EF1 and PO allocation for chores under additive valuations remains an open problem.

To address this issue, several studies have focused on restricted instances. The existence and polynomial-time computability of EF1 and fPO allocations for chores are known in the following cases: bi-valued instances~\cite{ebadian2022fairly, garg2022fair}; instances with two types of chores~\cite{Aziz2023}; instances with three agents~\cite{garg2023new}; and instances with two types of agents~\cite{garg2023new}.

\paragraph{Approximating Nash social welfare}

The development of approximation algorithms for maximizing Nash social welfare has been an active research topic in theoretical computer science in recent years, as maximum Nash social welfare (MNW) allocations satisfy both EF1 and PO~\cite{caragiannis2019unreasonable}.

For additive valuation functions, the first constant-factor approximation algorithm, achieving a ratio of $2\cdot \mathrm{e}^{1/\mathrm{e}}\approx 2.88$, was proposed by Cole and Gkatzelis ~\cite{cole2018approximating}. This factor was subsequently improved to $\mathrm{e}$~\cite{anari2017nash}, further refined to 2~\cite{cole2017convex}, and currently stands at the best-known approximation ratio of $\mathrm{e}^{1/\mathrm{e}} + \epsilon \approx 1.45$ ~\cite{barman2018finding}. Although approximating MNW is itself an interesting research topic, it is worth noting that an approximate MNW allocation does not necessarily satisfy EF1 or PO. In contrast, the 1.45-approximation algorithm by~\cite{barman2018finding} notably provides guarantees of approximate EF1 and PO, a significant achievement in this context.
In addition, similar approximation guarantees have been established for more general market models, such as piecewise-linear concave valuations~\cite{anari2018nash}, budget-additive valuations~\cite{garg2018approximating}, submodular valuations~\cite{garg2023approximating}, and multi-unit markets~\cite{bei2017earning}.

\subsection{Our Contributions}

The main contribution of this paper lies in proposing a polynomial-time algorithm to compute an allocation that achieves both EF1 and fPO under additive valuation functions when the number of agents is fixed.

\begin{theorem}\label{thm:main}
    When each agent has an additive valuation function and the number of agents is fixed, an EF1 and fPO allocation can be computed in polynomial time.
\end{theorem}

Moreover, our approach directly contributes to the Nash social welfare maximization problem. As mentioned above, Barman, Krishnamurthy, and Vaish~\cite{barman2018finding} developed a 1.45-approximation algorithm for maximizing Nash social welfare that achieves approximate EF1 and PO. We show that similar results hold: the EF1 and fPO allocation produced by our algorithm also serves as an $\mathrm{e}^{1/e} \approx 1.444$-approximation algorithm for the Nash social welfare maximization problem. Note that the Nash social welfare maximization problem is NP-hard even when there are only two agents. This result is significant as it provides theoretical guarantees for both fairness and efficiency while also approximating Nash social welfare.

\begin{theorem}\label{thm:nash}
    When each agent has an additive valuation function and the number of agents is fixed, there exists a polynomial-time $\mathrm{e}^{1/\mathrm{e}}$-approximation algorithm for the Nash social welfare maximization problem.
    Furthermore, the resulting allocation satisfies both EF1 and fPO.
\end{theorem}

\subsection{Our Techniques}

Our approach builds on the techniques introduced by Barman, Krishnamurthy, and Vaish~\cite{barman2018finding}. We begin by briefly outlining their method. They developed a pseudo-polynomial time algorithm to compute an allocation that is both EF1 and PO. Their algorithm first perturbs valuations to a desirable form, then computes an EF1 and fPO allocation for this perturbed instance. This resulting allocation is approximately EF1 and PO with respect to the original instance and becomes EF1 and PO if the perturbation is sufficiently small. 

Specifically, their algorithm maintains an integral allocation and prices for goods at each step, ensuring they correspond to an equilibrium outcome in a Fisher market. This equilibrium guarantees fPO by the first welfare theorem. The algorithm adjusts the allocation and prices iteratively by reallocating goods and increasing prices to approach a fairer allocation. The algorithm stops once the current allocation and prices achieve approximate {\it price envy-freeness up to one good } (pEF1).
Here, pEF1 requires that the spending of each agent is at least as high as that of any other agent after removing the most expensive good in the latter's bundle. Requiring the spending to be balanced in this manner yields EF1 for the corresponding fair division instance.

Our primary idea is to avoid processing the entire instance at once; instead, we sequentially add agents to the instance and construct an allocation that satisfies EF1 and fPO at each step. In the $k$-th iteration of the algorithm, we start with an EF1 and fPO allocation for an instance of $k-1$ agents, add the $k$-th agent to the instance, and find a new allocation that satisfies EF1 and fPO for $k$ agents. To do this, the algorithm maintains an allocation and prices for goods that correspond to an equilibrium outcome in a Fisher market.

In the $k$-th iteration, the algorithm achieves a PEF1 allocation by reallocating goods and adjusting prices. During this process, the allocation remains EF1 and fPO for the existing $k-1$ agents, while the reallocating and price increase are designed to eliminate the dissatisfaction of the newly added $k$-th agent.

Our approach differs from previous methods in several notable respects. First, we do not perturb the instance; rather, we compute an allocation that directly satisfies EF1 and fPO for the given instance. Second, in each iteration of our algorithm, we ensure that the {\it minimum spender}, who is the agent with the lowest spending, is always the newly added agent $k$. In previous algorithms, the minimum spender may change during reallocation, but our algorithm consistently operates to eliminate the dissatisfaction of $k$-th agent. This introduces a ``direction" for achieving fairness and enables new techniques to bound the number of iterations. Third, in the $k$-th iteration, since we need to maintain EF1 and fPO for the existing $k-1$ agents, we allow for simultaneous exchanges of multiple goods.
As far as we know, this approach is novel in the context of constructing fair and efficient allocations.

\subsection{Organization}
In Section~\ref{sec:pre}, we introduce the fair division model and the relevant notions of fairness and efficiency.
We also present the Fisher market framework and define key notions such as price envy-freeness, minimum spenders, maximum violators, and the MBB graph.
Section~\ref{sec:alg} gives a detailed description of the algorithms we propose.
In Section~\ref{sec:ana}, we analyze our algorithms and prove Theorem~\ref{thm:main}.
Finally, in Section~\ref{sec:con}, we summarize our results and suggest directions for future research.

Additionally, in Appendix~\ref{ap:1}, we provide the proof of Theorem~\ref{thm:nash}, and in Appendix~\ref{ap:2}, we discuss an error of the proof in~\cite{garg2024computing}.

\section{Preliminaries}\label{sec:pre}
For positive integer $q$, let $[q]$ denote $\{1,\ldots, q\}$.
\paragraph{The fair division model}
A {\it fair division instance} is represented by a tuple $(N,M,V)$, where $N=[n]$ denotes a set of $n$ agents, $M=[m]$ denotes a set of $m$ indivisible goods, and $V=\{v_i\}_{i\in N}$ represents a set of valuation functions of each agent $i\in N$.
In this paper, we assume that each valuation function $v_i: 2^M \rightarrow \mathbb{R}_{\ge 0}$ is {\it additive}, i.e., $v_i(S)=\sum_{g\in S} v_i(\{g\})$ for each $i\in N$ and any $S\subseteq M$.\footnote{We assume that $v_i(\emptyset)=0$ for all $i\in N$.}
To simplify notation, we will write $v_{ig}$ instead of $v_i(\{g\})$ for a singleton good $g\in M$.
We assume that for each good $g \in M$, there exists an agent $i \in N$ such that $v_{ig}>0$.
Otherwise, we eliminate the good from the instance.
Similarly, we assume that for each agent $i \in N$, there exists a good $g \in M$ such that $v_{ig}>0$.
Otherwise, we eliminate the agent from the instance.

An {\it allocation} $\bx=(\bx_i)_{i \in N}$ is an $n$-partition of $M$, where $\bx_i\subseteq M$ is the {\it bundle} allocated to agent $i$.
Given an allocation $\bx$, the valuation of agent $i\in N$ for the bundle $\bx_i$ is $v_i(\bx_i)=\sum_{g\in \bx_i} v_{ig}$.
A {\it fractional allocation} $\bx = (\bx_i)_{i \in N}$ represents a fractional assignment of the goods to the agents, where each $\bx_i$ is a vector $(x_{ig})_{g \in M}$, and $x_{ig} \in [0,1]$ denotes the fraction of good $g$ allocated to agent $i$. Additionally, this allocation satisfies the condition that $\sum_{i \in N} x_{ig} \le 1$ for each good $g \in M$.
Given a fractional allocation $\bx$, the valuation of agent $i\in N$ for $\bx_i$ is $v_i(\bx_i)=\sum_{g\in M} x_{ig}\cdot v_{ig}$.
We will use ``allocation'' to mean an integral allocation and specify ``fractional allocation'' otherwise.

\paragraph{Fairness notions}
Given a fair division instance $(N,M,V)$ and an allocation $\bx=(\bx_i)_{i \in N}$, 
we say that an agent $i\in N$ {\it envies} another agent $j\in N$ if $v_i(\bx_i)<v_i(\bx_j)$.
An allocation $\bx$ is said to be {\it envy-free} (EF) if no agent envies any other agent.
An allocation $\bx$ is said to be {\it envy-free up to one good} (EF1) if for any pair of agents $i,j\in N$ where $i$ envies $j$\footnote{The phrase ``where $i$ envies $j$'' is necessary; otherwise, the condition is not satisfied when $\bx_j$ is an empty set.}, there exists a good $g\in \bx_j$ such that $v_i(\bx_i)\ge v_i(\bx_j\setminus \{g\})$.

\paragraph{Efficiency notions}
Given a fair division instance $(N,M,V)$ and an allocation $\bx=(\bx_i)_{i \in N}$, 
we say that $\bx$ is {\it Pareto dominated} by another allocation $\by$ if $v_i(\by_i)\ge v_i(\bx_i)$ for every agent $i\in N$, and $v_j(\by_j)> v_j(\bx_j)$ for some agent $j\in N$.
An allocation $\bx$ is said to be {\it Pareto optimal} (PO) if $\bx$ is not Pareto dominated by any other allocation.
Similarly, an allocation $\bx$ is said to be {\it fractionally Pareto optimal} (fPO) if $\bx$ is not Pareto dominated by any fractional allocation.
Note that a fractionally Pareto optimal allocation is also Pareto optimal, but not vice versa (see the full version of \cite{barman2018finding} for an example).

\paragraph{Nash social welfare}
Given an allocation $\bx$, the {\it Nash social welfare} of $\bx$ is defined as $(\Pi_{i\in N} v_i(\bx_i))^{\frac{1}{n}}$.

\paragraph{Fisher market}
The Fisher market is a fundamental model of an economic market introduced by Fisher in 1891~\cite{brainard2005compute}.
In the {\it linear Fisher market model}, 
an instance is represented by a tuple $(N,M,V,B)$, where $N=[n]$ denotes a set of $n$ {\it buyers}, and $M=[m]$ denotes a set of $m$ {\it divisible goods}, each of which is assumed to be a unit without loss of generality. 
$V=\{v_i\}_{i\in N}$ represents a set of {\it additive} valuation functions $v_i: 2^M \rightarrow \mathbb{R}_{\ge 0}$ of each buyer $i\in N$, and $B=\{b_i\}_{i\in N}$ represents a set of {\it budgets} of each buyer $i\in N$, where each $b_i$ is a non-negative real number.
We refer to the tuple $(N,M,V,B)$ as a {\it Fisher market instance}.

A {\it market outcome} is represented by the pair $(\bx, \bp)$, where $\bx=(\bx_i)_{i \in N}$ is a fractional allocation of $m$ goods, and $\bp=(p_g)_{g\in M}$ is a {\it price vector} representing the {\it prices} for each good $g\in M$.
Here, each price $p_g$ is a non-negative real number.
The {\it spending} of buyer $i$ under the market outcome $(\bx, \bp)$ is defined as $\bp(\bx_i)=\sum_{g\in M}x_{ig}\cdot p_{g}$.
The {\it valuation} of buyer $i$ under the market outcome $(\bx, \bp)$ is given by $v_i(\bx_i)=\sum_{g\in M} x_{ig}\cdot v_{ig}$.

Given a price vector $\bp=(p_g)_{g\in M}$, the {\it bang-per-buck ratio} for buyer $i$ with respect to good $g$ is defined as $v_{ig}/p_g$\footnote{If $v_{ig}=p_g=0$, then we define $v_{ig}/p_g=0$.}.
The {\it maximum bang-per-buck ratio} (MBB ratio) for buyer $i$ is given by $\alpha_i:=\max_{g\in M}v_{ig}/p_g$.
Let ${\rm MBB}_i:=\{g\in M\mid v_{ig}/p_g=\alpha_i\}$ denote the set of goods that maximize the bang-per-buck ratio for buyer $i$.
We refer to ${\rm MBB}_i$ as the {\it MBB set} for buyer $i$.

An outcome $(\bx, \bp)$ is said to be an {\it equilibrium} if it satisfies the following conditions:
\begin{itemize}
    \item {\it Market clearing}:
    Each good is either priced at zero or is allocated fully, i.e., 
    $\forall g\in M$, $p_g=0$ or $\sum_{i\in N} x_{ig}=1$.
    \item {\it Budget exhaustion}:
    Each buyer spends all of their budget, i.e., 
    $\forall i\in N, \bp(\bx_i)=b_i$.
    \item {\it Maximum bang-per-buck constraints}:
    Each buyer spends their budget only on their MBB set, i.e., 
    $\forall i\in N, \forall g\in M$, $x_{ig}>0 \implies g\in {\rm MBB}_i$.
\end{itemize}
An equilibrium outcome always exists under the assumption that for each good $j$, there exists an agent $i$ such that $v_{ij} > 0$, and every buyer values at least one good above $0$~\cite[Chapter 5]{NisaRougTardVazi07}.

The First Welfare Theorem, a fundamental and important theorem in economics, particularly in general equilibrium theory, implies that the equilibrium outcome in a Fisher market is fractionally Pareto optimal.
\begin{proposition}(First Welfare Theorem;~\cite[Chapter 16]{mas1995microeconomic})\label{prop:first}
For a linear Fisher market model, any equilibrium outcome is fractionally Pareto optimal (fPO).
\end{proposition}

We similarly apply the concepts of price vectors and MBB sets within the framework of fair division problems.
For a given fair division instance $(N,M,V)$, we refer to the pair $(\bx, \bp)$ as a {\it solution}, where $\bx=(\bx_i)_{i\in N}$ is an allocation, and $\bp=(p_g)_{g\in M}$ is a price vector.

\paragraph{Price envy-freeness, minimum spenders, maximum violators}
Price envy-freeness (pEF), price envy-free up to one good (pEF1), minimum spenders, and violators are fundamental concepts introduced in \cite{barman2018finding}.
In this paper, we introduce a new concept termed the {\it maximum violator}, which plays a crucial role in our discussion.

Let $(N,M,V)$ be a fair division instance and let $(\bx, \bp)$ be a solution in $(N,M,V)$.
We say that an agent $i\in N$ {\it price envies} another agent $j\in N$ if $\bp(\bx_i)<\bp(\bx_j)$.
A solution $(\bx, \bp)$ is said to be {\it price envy-free} (pEF) if no agent price envies any other agent.
A solution $(\bx, \bp)$ is said to be {\it price envy-free up to one good} (pEF1) if for any pair of agents $i,j\in N$ where $i$ price envies $j$, there exists a good $g\in \bx_j$ such that $\bp(\bx_i)\ge \bp(\bx_j\setminus \{g\})$.

For any $S\subseteq M$, let $\hat{\bp}(S)$ denote the total price after removing the highest-priced good from $S$. Formally, $\hat{\bp}(S)$ is defined as follows:
\begin{align*}
  \hat{\bp}(S):=
  \left\{
    \begin{array}{ll}
      \min_{g\in S} \bp(S\setminus \{g\}) & {\rm if}~S \neq \emptyset,\\
      0 & {\rm otherwise}.
    \end{array}
  \right.
\end{align*}
Let us denote $\min(\bx, \bp):= \min_{i \in N} \bp(\bx_i)$ and $\max(\bx, \hat{\bp}):=\max_{i\in N}\hat{\bp}(\bx_i)$.
\begin{observation}\label{ob:pEF1}
A solution $(\bx, \bp)$ is pEF1 if and only if $\min(\bx, \bp)\ge \max(\bx, \hat{\bp})$ holds.
\end{observation}
\begin{proof}
From the definition, we directly obtain 
    \begin{align*}
    (\bx, \bp)~\text{is pEF1}
     &\Longleftrightarrow \forall i, j \in N, \bp(\bx_i)\ge \hat{\bp}(\bx_j)  \\
     &\Longleftrightarrow \min(\bx, \bp)\ge \max(\bx, \hat{\bp}).
    \end{align*}
\end{proof}
We define {\it minimum spender set} $L(\bx, \bp):=\{i \in N\mid \bp(\bx_i)=\min(\bx, \bp)\}$ and
{\it maximum violator set} $K(\bx, \bp):=\{i \in N\mid \hat{\bp}(\bx_i)=\max(\bx, \hat{\bp})\}$ for a solution $(\bx, \bp)$.
When it is clear from the context, $L(\bx, \bp)$ and $K(\bx, \bp)$ may sometimes be abbreviated as $L$ and $K$.
\paragraph{MBB graph and augmented MBB graph}
For a given fair division instance $(N,M,V)$ and a price vector $\bp$, define the {\it MBB graph} as a directed bipartite graph $G(\bp)$ with an edge $(i,g)$ from agent $i$ to good $g$ if and only if $g$ is in the MBB set of agent $i$ at prices $\bp$, i.e., $g\in {\rm MBB}_i$.
We call an edge $(i,g)$ an {\it MBB edge}.
In addition, given an allocation $\bx$, define the {\it augmented MBB graph} $G(\bx, \bp)$ as a directed bipartite graph obtained from $G(\bp)$ by adding an edge $(g,i)$ from good $g$ to agent $i$, where $g\in \bx_i$.
We call an edge $(g,i)$ an {\it allocation edge}.
\begin{example}
Consider a fair division instance $(N,M,V)$ with three agents $\{1,2,3\}$ and five goods $\{1, 2, 3, 4, 5\}$. The values are given in Table~\ref{tb:1}.
\begin{table}[htbp]
    \centering
    \begin{tabular}{|c||c|c|c|c|c|}
        \hline
        & $1$ & $2$ & $3$ & $4$ & $5$ \\
        \hline 
        agent $1$ & 6 & 5 & 0 & 0 & 0 \\
        agent $2$ & 0 & 1 & 7 & 3 & 0 \\
        agent $3$ & 2 & 3 & 6 & 3 & 4 \\
        \hline
    \end{tabular}
    \caption{A fair division instance}
    \label{tb:1}
\end{table}

Consider the allocation $\bx=(\bx_i)_{i\in [3]}$ given by $\bx_1=\{1,2\}$, $\bx_2=\{3,4\}$, $\bx_3=\{5\}$ and the price vector $\bp=(6,5,7,3,4)$.
Figure~\ref{fig:mbb} describes the MBB graph $G(\bp)$ and the augmented MBB graph $G(\bx, \bp)$.
The minimum spender is agent $3$ and the maximum violator is agent $1$ in this situation.
\begin{figure}[h]
    \centering
    \includegraphics[width=0.7\textwidth]{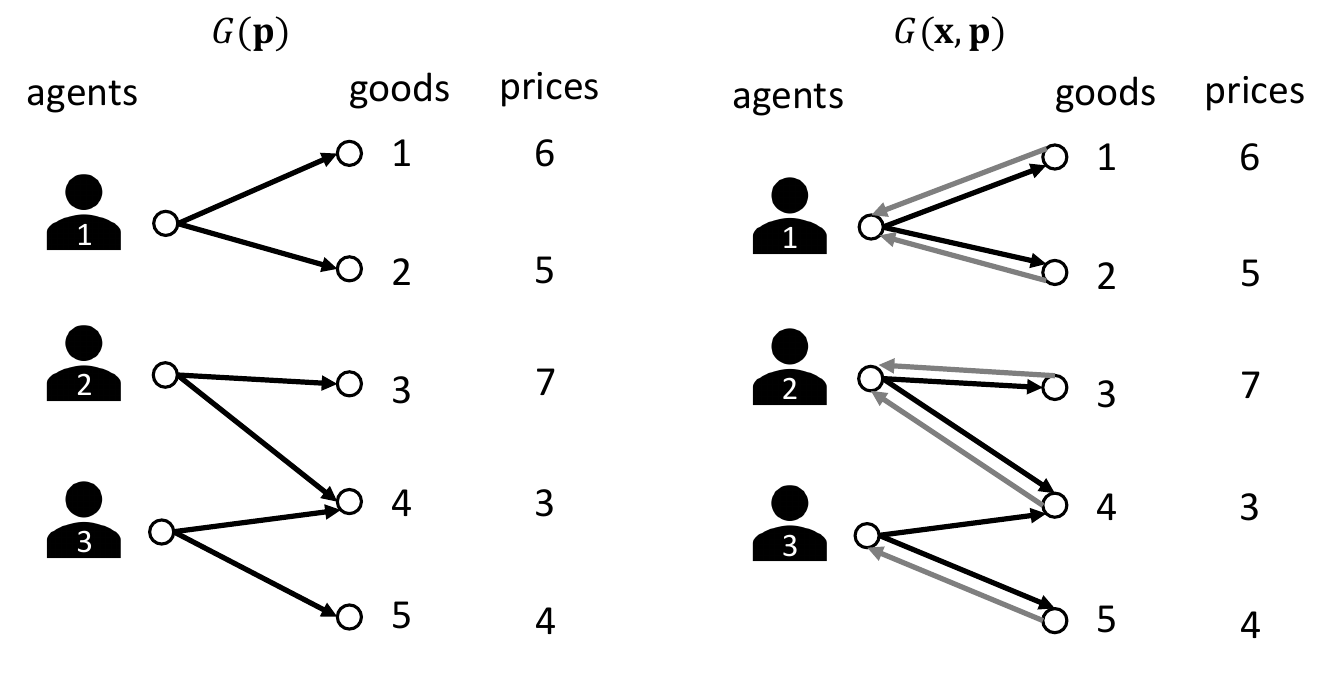}
    \caption{MBB graph $G(\bp)$ and augmented MBB graph $G(\bx, \bp)$. In $G(\bx,\bp)$, the black edges represent MBB edges, while the gray edges represent allocation edges.}
    \label{fig:mbb}
\end{figure}
\end{example}

The following lemma provides a sufficient condition for finding an allocation $\bx$ that satisfies both EF1 and fPO.
\begin{lemma}\cite{barman2018finding}\label{lem:pEF1}
Given a fair division instance $(N,M,V)$, 
let $(\bx, \bp)$ be a solution such that (1) $(\bx, \bp)$ is pEF1, and (2) $\bx_i\subseteq {\rm MBB}_i$ for each $i\in N$.
Then, $\bx$ is EF1 and fPO allocation in $(N,M,V)$.
\end{lemma}

\paragraph{Instances that satisfy Hall's condition}

We assume that the given fair division instance satisfies {\it Hall's condition}~\cite{hall1987representatives}, as used in~\cite{barman2018finding}. More precisely, for a given fair division instance $(N, M, V)$, we define the {\it valuation graph} $\mathcal{G} = (N, M; E)$ as an unweighted bipartite graph, where $E := \{(i, g) \mid v_{ig} > 0 \}$. The instance $(N, M, V)$ is said to {\it satisfy Hall's condition} if, for any subset $A \subseteq N$, $|A| \le |\Gamma(A)|$ holds, where $\Gamma(A)$ denotes the set of goods valued positively by at least one agent in $A$, i.e., $\Gamma(A) := \{g\in M \mid \exists i \in N ~{\rm s.t.}~ (i,g)\in E \}$.

If the given instance does not satisfy Hall's condition, we apply an appropriate preprocessing step to reduce it to a smaller instance that does satisfy Hall's condition. 
Note that this transformation can be computed in polynomial time in $n$ and $m$.
For further details, we refer to the full version of~\cite{barman2018finding}.
From now on, we assume that the given fair division instance satisfies Hall's condition.

\section{Our Algorithms}\label{sec:alg}
In this section, we provide a detailed description of our algorithms.
\paragraph{Main algorithm}
Our main algorithm (Algorithm \ref{alg:1}) computes an EF1 and fPO allocation $\bx$ for a given fair division instance $(N,M,V)$.
To achieve this, the algorithm constructs an allocation $\bx=(\bx_i)_{i\in N}$ and a price vector $\bp=(p_g)_{g\in M}$ such that (1) $(\bx, \bp)$ is pEF1, and (2) $\bx_i\subseteq {\rm MBB}_i$ for each $i\in N$.
By Lemma~\ref{lem:pEF1}, we conclude that $\bx$ is an EF1 and fPO allocation for the instance.

The algorithm constructs a solution $(\bx, \bp)$ that satisfies these conditions by sequentially adding agents to the current instance and modifying the solution accordingly at each step.
Initially, the set of agents $\mathcal{N}$, the set of goods $\mathcal{M}$, and the set of valuation functions $\mathcal{V}$ are empty. At the $k$-th iteration, the allocation $\bx_k$ is set to $M_k$, where $M_k$ denotes the set of goods not yet included in the current instance that are positively valued by agent $k$, i.e., $M_k := \{g\in M\setminus \mathcal{M} \mid v_{kg} > 0\}$.

For each good $g\in M_k$, the price $p_g$ is set to $\frac{v_{kg}\cdot \min_{h\in \mathcal{M}} p_h}{m\cdot \max_{h\in M} v_{kh}}$\footnote{If $\mathcal{M}=\emptyset$, we define $\min_{h\in \mathcal{M}} p_h$ to be $1$.}.
Intuitively, these prices are set low enough to ensure that $M_k$ is included in the MBB set for agent $k$, and none of the agents already added to the instance price envy agent $k$. For further details, refer to the proof of Theorem~\ref{thm:correctness}.
The algorithm then updates the current instance by adding agent $k$, the set of goods $M_k$, and the valuation function $v_k$.

After this, the algorithm calls the \textit{FindSolution} algorithm (Algorithm~\ref{alg:2}) to compute a new solution $(\bx,\bp)$ for the current instance $(\mathcal{N}, \mathcal{M}, \mathcal{V})$ that satisfies the following two conditions: (1) $(\bx, \bp)$ is pEF1 in $(\mathcal{N}, \mathcal{M}, \mathcal{V})$, and (2) $\bx_i\subseteq {\rm MBB}_i$ for each $i\in {\mathcal{N}}$.

After all iterations, the algorithm returns the final allocation $\bx=(\bx_i)_{i\in \mathcal{N}}$. 
\begin{remark}
    Readers may wonder why the algorithm incrementally adds the set of goods $M_k$ to the instance, rather than setting $\mathcal{M} = M$ from the beginning. The reason is to ensure that all goods are assigned non-zero prices. This is crucial for guaranteeing that the price increase rate in the following \textit{FindSolution} algorithm is bounded, i.e., $\beta < \infty$.
\end{remark}
The following observation will be used to ensure that the price increase  rate in the \textit{FindSolution} algorithm is bounded.
\begin{observation}\label{ob:hall}
    At the beginning of each iteration in Algorithm~\ref{alg:1},  
    the current instance $(\mathcal{N}, \mathcal{M}, \mathcal{V})$ satisfies Hall's condition.
\end{observation}
\begin{proof}
At the beginning of the first iteration, since $\mathcal{N}$ and $\mathcal{M}$ are empty sets, it is clear that $(\mathcal{N}, \mathcal{M}, \mathcal{V})$ satisfies Hall's condition. 
Fix any $k \in [n-1]$ and consider the beginning of the $k+1$-th iteration of Algorithm~\ref{alg:1}. 
At this point, we have $\mathcal{N} = [k]$, $\mathcal{M} = \cup_{j=1}^k M_j$, and $\mathcal{V} = \{v_j\}_{j \in [k]}$. 
Let $\mathcal{G} = (N, M; E)$ be the valuation graph for the instance $(N, M, V)$, and let $\mathcal{G}' = (\mathcal{N}, \mathcal{M}; \mathcal{E})$ be the valuation graph for $(\mathcal{N}, \mathcal{M}, \mathcal{V})$. 
Since all goods with positive value for any agent in $[k]$ are included in $\mathcal{M}$, it follows that $\forall i \in [k], \forall g \in M$, $(i, g) \in E \implies (i, g) \in \mathcal{E}$. 
Therefore, since $(N, M, V)$ satisfies Hall's condition by the assumption, $(\mathcal{N}, \mathcal{M}, \mathcal{V})$ also satisfies Hall's condition.

\end{proof}
\begin{algorithm}[tb]
\caption{Main Algorithm}
\label{alg:1}
\begin{algorithmic}[1]
\Require Fair division instance $(N, M, V)$
\Ensure An EF1 and fPO allocation $\bx$
\State $\mathcal{N}\leftarrow \emptyset$, $\mathcal{M}\leftarrow \emptyset$ $\mathcal{V}\leftarrow \emptyset$
\Comment{The current set of agents, goods, and valuation functions}
\For{$k=1, 2,\cdots, n$}
\State Let $M_k = \{g\in M\setminus \mathcal{M} \mid v_{kg} > 0\}$.
\State $\bx_k \leftarrow M_k$, and $p_g \leftarrow \frac{v_{kg}\cdot \min_{h\in \mathcal{M}} p_h}{m\cdot \max_{h\in M} v_{kh}}$ for any $g \in M_k$
\State $\mathcal{N}\leftarrow \mathcal{N}\cup \{k\}$, $\mathcal{M}\leftarrow \mathcal{M}\cup M_k$, $\mathcal{V}\leftarrow \mathcal{V}\cup \{v_k\}$
\State $(\bx,\bp)\leftarrow$ {\it FindSolution}$(\mathcal{N}, \mathcal{M}, \mathcal{V},\bx, \bp)$

\EndFor
\Return $\bx$
\end{algorithmic}
\end{algorithm}

\paragraph{The \textit{FindSolution} algorithm}
The \textit{FindSolution} algorithm (Algorithm \ref{alg:2}) modifies the current allocation and price vector in a fair division instance $(\mathcal{N}, \mathcal{M}, \mathcal{V})$ to ensure that the resulting solution $(\bx,\bp)$ satisfies the following two conditions: (1) $(\bx, \bp)$ is pEF1, and (2) $\bx_i\subseteq {\rm MBB}_i$ for each $i\in \mathcal{N}$.

Let $k$ be the most recently added agent to the set $\mathcal{N}$.
The algorithm maintains the invariant that $(\bx, \bp)$ is pEF1 for all agents except $k$ and that $\bx_i \subseteq {\rm MBB}_i$ for each $i\in \mathcal{N}$.
More precisely, the first condition means that $\bp(\bx_i)\ge \max(\bx,\hat{\bp})$ holds for any $i \in \mathcal{N}\setminus \{k\}$ by Observation~\ref{ob:pEF1}.

The algorithm iterates while the current solution $(\bx,\bp)$ is not pEF1 in $(\mathcal{N}, \mathcal{M}, \mathcal{V})$.
Let $L$ be the minimum spender set and let $K$ be the maximum violator set for the current solution $(\bx, \bp)$ in $(\mathcal{N}, \mathcal{M}, \mathcal{V})$, i.e., $L=\{i \in \mathcal{N} \mid \bp(\bx_i)=\min(\bx, \bp)\}$ and $K=\{i \in \mathcal{N} \mid \hat{\bp}(\bx_i)=\max(\bx, \hat{\bp})\}$.
Note that, as shown later in Corollary~\ref{cor:L={k}}, since all agents except $k$ satisfy the pEF1 conditions and $(\bx,\bp)$ is not pEF1, the minimum spender set must be $L=\{k\}$.

We define $R_{\mathcal{N}}$ as the set of agents reachable from $L$ via  directed paths in the augmented MBB graph $G(\bx, \bp)$, and $R_{\mathcal{M}}$ as the set of goods reachable from $L$ via directed paths in $G(\bx, \bp)$.

If $K\cap R_{\mathcal{N}}\neq \emptyset$, meaning that there exists at least one maximum violator that can be reached from $L$ via directed paths in $G(\bx, \bp)$, the algorithm calls the \textit{Transfer} algorithm (Algorithm~\ref{alg:3}) to modify the current allocation $\bx$.

Otherwise, the algorithm uniformly increases the prices of all goods in $R_{\mathcal{M}}$ until one of the following events occurs: (i) an MBB edge appears between $R_{\mathcal{N}}$ and $\mathcal{M} \setminus R_{\mathcal{M}}$ in $G(\bx,\bp)$, (ii) an agent in $R_{\mathcal{N}}$ becomes a maximum violator, or (iii) the solution $(\bx,\bp)$ satisfies pEF1.

Let $\beta_1$, $\beta_2$, and $\beta_3$ denote the price increase rates corresponding to the first, second, and third events, respectively. These rates are defined as follows:
\begin{flalign*}
    \beta_1 &:= \min_{j\in R_{\mathcal{N}}, g\in \mathcal{M} \setminus R_{\mathcal{M}}} \frac{p_g\cdot \alpha_j}{v_{jg}}, \\
    \beta_2 &:= \min_{j\in R_{\mathcal{N}}} \frac{\max(\bx, \hat{\bp})}{\hat{\bp}(\bx_j)}, \\
    \beta_3 &:= \frac{\max(\bx, \hat{\bp})}{\bp(\bx_k)}.
\end{flalign*}
The overall price increase rate $\beta$ is given by $\beta = \min \{\beta_1, \beta_2, \beta_3\}$.
The algorithm multiplies the prices of all goods in $R_{\mathcal{M}}$ by $\beta$.
The following lemma guarantees that the prices of all goods is always positive, finite and non-decreasing.
\begin{lemma}\label{lem:price}
    Throughout the execution of Algorithm~\ref{alg:1}, the price of all goods is positive, finite, and non-decreasing.
\end{lemma}
\begin{proof}
Note that Price changes occur only in line 10 of the {\it FindSolution} algorithm.
First, we show that if $1 < \beta < \infty$ holds in Algorithm~\ref{alg:2}, then the initial price of each good must be positive. Indeed, since we assume that the instance $(N,M,V)$ satisfies Hall's condition, we have $\max_{h \in M} v_{ih} > 0$ for each $i \in N$. In the initial iteration of Algorithm~\ref{alg:1}, the price of each $g \in M_1$ is set as $p_g = \frac{v_{1g}}{m \cdot \max_{h \in M} v_{1h}}$, which is positive since $v_{1g} > 0$. Since $1 < \beta < \infty$ holds in Algorithm~\ref{alg:2}, we can conclude inductively that $\min_{h \in \mathcal{M}} p_h > 0$ at the beginning of each iteration of Algorithm~\ref{alg:1}, ensuring that the initial price of each good must be positive.

Now, we show that $1<\beta <\infty$ holds in Algorithm~\ref{alg:2}.
First, we show that $\beta_1, \beta_2$ and $\beta_3$ are well-defined and greater than $1$.
Since $L \neq \emptyset$ and $L\subseteq R_{\mathcal{N}}$, we have $R_{\mathcal{N}} \neq \emptyset$.
Furthermore, since $K \neq \emptyset$ and $K \cap R_{\mathcal{N}} = \emptyset$, we have $K\subseteq \mathcal{N} \setminus R_{\mathcal{N}}$, which implies that $\mathcal{N} \setminus R_{\mathcal{N}} \neq \emptyset$.  
Additionally, since $(\bx, \bp)$ is not pEF1, the maximum violators must have at least one good, implying that $\mathcal{M} \setminus R_{\mathcal{M}} \neq \emptyset$.

From the definition of $R_{\mathcal{M}}$, we observe that there are no MBB edges between $R_{\mathcal{N}}$ and $\mathcal{M} \setminus R_{\mathcal{M}}$.  
This means that for any $j \in R_{\mathcal{N}}$ and any $g \in \mathcal{M} \setminus R_{\mathcal{M}}$, $v_{jg}/p_g < \alpha_j$ holds, where $\alpha_j$ is MBB ratio for agent $j$. 
Consequently, $\beta_1$ is well-defined and satisfies $\beta_1 > 1$.

Next, since $R_{\mathcal{N}} \neq \emptyset$, $\beta_2$ is well-defined.
As $K \cap R_{\mathcal{N}} = \emptyset$, it follows that $\hat{\bp}(\bx_j) < \max(\bx, \hat{\bp})$ for any $j\in R_{\mathcal{N}}$, implying  $\beta_2 > 1$.  

Finally, $\beta_3$ is obviously well-defined.
As $L=\{k\}$, we have $\bp(\bx_k)< \max(\bx, \hat{\bp})$, implying that $\beta_3 > 1$.
Combining the above results, we have $\beta > 1$.

We next show that $\beta < \infty$ by considering the following three cases:
\begin{itemize}
    \item There exists an agent $j \in R_{\mathcal{N}}$ who holds more than one good.\\
    Since the prices of all goods are positive, we have $\hat{\bp}(\bx_j) > 0$, implying $\beta_2 < \infty$, and hence $\beta < \infty$.
    \item Every agent $j \in R_{\mathcal{N}}$ holds exactly one good.\\
    Since $\bp(\bx_k) > 0$, it follows that $\beta_3 < \infty$, and hence $\beta < \infty$.
    \item Otherwise.\\
    In this case, every agent in $R_{\mathcal{N}}$ holds at most one good, and there exists an agent $j \in R_{\mathcal{N}}$ who holds no goods.
    By Observation~\ref{ob:hall}, the instance $(\mathcal{N}, \mathcal{M}, \mathcal{V})$ satisfies Hall’s condition. Thus, there exists an agent $j \in R_{\mathcal{N}}$ and a good $g \in \mathcal{M} \setminus R_{\mathcal{M}}$ such that $v_{jg} > 0$ by Hall’s condition. This implies that $\beta_1 < \infty$, and hence $\beta < \infty$.
\end{itemize}
Thus, in all cases, $1 < \beta < \infty$ holds.
This completes the proof of Lemma~\ref{lem:price}.
\end{proof}

\begin{algorithm}[htb]
\caption{{\it FindSolution}$(\mathcal{N}, \mathcal{M}, \mathcal{V},\bx, \bp)$}
\label{alg:2}
\begin{algorithmic}[1]
\Require A fair division instance $(\mathcal{N}, \mathcal{M}, \mathcal{V})$, and a solution $(\bx, \bp)$ in this instance such that $(\bx, \bp)$ is pEF1 except for the most recently added agent, and $\bx_i\subseteq {\rm MBB}_i$ for each $i\in \mathcal{N}$.
\Ensure A solution ($\bx'$, $\bp'$) such that $(\bx', \bp')$ is pEF1 and $\bx_i\subseteq {\rm MBB}_i$ for each $i\in \mathcal{N}$.
\While{the current solution $(\bx,\bp)$ is not pEF1}
\newline\Comment{Invariant: $(\bx, \bp)$ is pEF1 except for the most recently added agent, and $\bx_i\subseteq {\rm MBB}_i$ for each $i\in \mathcal{N}$.}
\If{$K\cap R_{\mathcal{N}}\neq \emptyset$}
\State $\bx \leftarrow$ {\it Transfer}$(\mathcal{N}, \mathcal{M}, \mathcal{V},\bx, \bp)$
\Else
\State $\beta_1 \leftarrow \min_{j\in R_{\mathcal{N}},g\in \mathcal{M} \setminus R_{\mathcal{M}}} \frac{p_g\cdot \alpha_j}{v_{jg}}$
\newline\Comment{$\beta_1$ corresponds to increasing prices until an MBB edge appears between $R_{\mathcal{N}}$ and $\mathcal{M} \setminus R_{\mathcal{M}}$.}
\State $\beta_2 \leftarrow \min_{j\in R_{\mathcal{N}}} \frac{\max(\bx, \hat{\bp})}{\hat{\bp}(\bx_j)}$
\newline\Comment{$\beta_2$ corresponds to increasing prices until an agent in $R_{\mathcal{N}}$ becomes a maximum violator.}
\State $\beta_3 \leftarrow \frac{\max(\bx, \hat{\bp})}{\bp(\bx_k)}$, where $k$ is the most recently added agent.
\newline\Comment{$\beta_3$ corresponds to increasing prices until $(\bx,\bp)$ satisfies pEF1.}
\State $\beta \leftarrow \min \{\beta_1, \beta_2, \beta_3\}$
\For{$\forall g \in R_{\mathcal{M}}$}
\State $p_g\leftarrow \beta \cdot p_g$
\EndFor
\EndIf
\EndWhile
\Return $(\bx, \bp)$
\end{algorithmic}
\end{algorithm}

\paragraph{The \textit{Transfer} algorithm}
The \textit{Transfer} algorithm (Algorithm \ref{alg:3}) modifies the current allocation $\bx$ to a new allocation $\bx'$ in the instance $(\mathcal{N}, \mathcal{M}, \mathcal{V})$, while preserving the invariant that $(\bx', \bp)$ is pEF1 except for the most recently added agent and that $\bx'_i\subseteq {\rm MBB}_i$ for each $i\in \mathcal{N}$.

The algorithm finds a directed path $P=(i_0, g_1, i_1, g_2,\ldots ,i_{\ell-1}, g_\ell, i_\ell)$ of the {\it shortest} length in the augmented MBB graph $G(\bx, \bp)$, where $i_0\in L$ and $i_\ell \in K$.
Note that, as shown later in Corollary~\ref{cor:L={k}}, the minimum spender set must be $L=\{k\}$.
Since the \textit{Transfer} algorithm is called only when the condition $K\cap R_{\mathcal{N}}\neq \emptyset$ holds, there is guaranteed to be at least one such directed path $P$.

Next, it identifies two critical indices along the path $P$. The first index, $a \in [\ell]$, is the smallest index such that $\bp(\bx_{i_a}\setminus \{g_a\})\ge \max(\bx, \hat{\bp})$. This means that $i_a$ is the closest agent to $i_0$ on the path $P$ that satisfies this condition.
Since $i_\ell \in K$, we have $\bp(\bx_{i_\ell}\setminus \{g_\ell \})\ge \max(\bx, \hat{\bp})$, and thus there exists an index $a\in [\ell]$ that satisfies this condition. 

The second index, $b\in [a-1]$, is the largest index prior to $a$ such that $\max(\bx, \hat{\bp}) \ge \bp(\bx_{i_b}\cup \{g_{b+1}\} \setminus \{g_b\})$.
This means that $i_b$ is the closest agent to $i_a$ among the agents between $i_1$ and $i_{a-1}$ that satisfies this condition.
If no such $b$ exists, it is set to $0$.

The algorithm then reallocates the goods along the identified path.
Specifically, agent $i_a$ relinquishes $g_a$, agent $i_b$ acquires $g_{b+1}$, and each intermediate agent $i_c$ relinquishes $g_c$ and acquires $g_{c+1}$.
The allocations for the remaining agents remain unchanged.

Finally, the algorithm returns the updated allocation $\bx'$.
Figure~\ref{fig:transfer} describes this reallocation process along the identified path in the augmented MBB graph $G(\bx, \bp)$.

\begin{figure}[t!]
    \centering
    \includegraphics[width=0.6\textwidth]{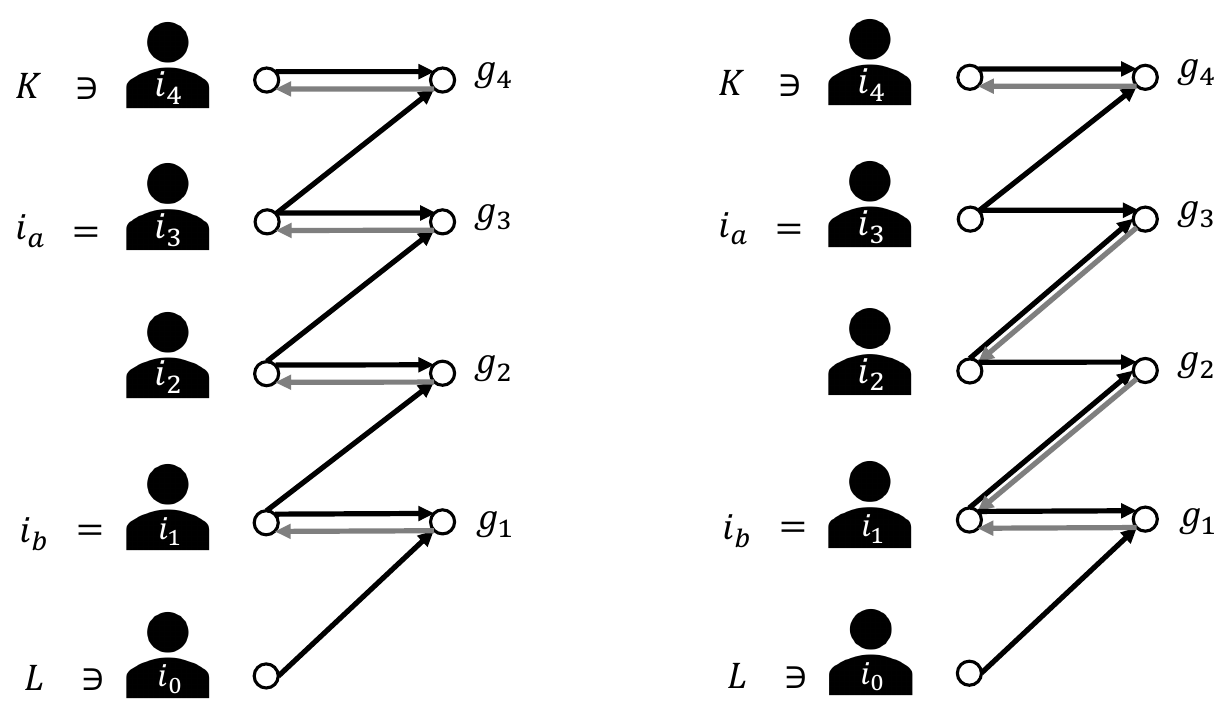}
    \caption{Illustration of the Transfer algorithm (Algorithm~\ref{alg:3}). The left figure shows the state before applying the Transfer algorithm, and the right figure shows the state after applying it. The black edges represent MBB edges, while the gray edges represent allocation edges in the augmented MBB graph $G(\bx, \bp)$.}
    \label{fig:transfer}
\end{figure}

\begin{algorithm}[h!]
\caption{{\it Transfer}$(\mathcal{N}, \mathcal{M}, \mathcal{V},\bx, \bp)$}
\label{alg:3}
\begin{algorithmic}[1]
\Require A fair division instance $(\mathcal{N}, \mathcal{M}, \mathcal{V})$, and a solution $(\bx, \bp)$ in this instance such that $(\bx, \bp)$ is pEF1 except for the most recently added agent, and $\bx_i\subseteq {\rm MBB}_i$ for each $i\in \mathcal{N}$.
\Ensure An allocation $\bx'=(\bx'_i)_{i \in \mathcal{N}}$ such that $(\bx', \bp)$ is pEF1 except for the most recently added agent, and $\bx_i\subseteq {\rm MBB}_i$ for each $i\in \mathcal{N}$.
\State Let $L$ be the set of minimum spenders and $K$ the set of maximum violators for $(\bx, \bp)$.
\State Find a shortest directed path $P=(i_0, g_1, i_1, g_2,\ldots ,i_{\ell-1}, g_\ell, i_\ell)$ in the augmented MBB graph $G(\bx, \bp)$, where $i_0\in L$ and $i_\ell \in K$.
\State Let $a \in [\ell]$ be the smallest index such that $\bp(\bx_{i_a}\setminus \{g_a\})\ge \max(\bx, \hat{\bp})$.
\State Let $b \in [a-1]$ be the largest index such that $\max(\bx, \hat{\bp}) \ge \bp(\bx_{i_b}\cup \{g_{b+1}\} \setminus \{g_b\})$. (If no such $b$ exists, set $b=0$.)
\For{$\forall i \in \mathcal{N}$}
\If{$i=i_a$}
\State $\bx'_{i_a}\leftarrow \bx_{i_a} \setminus \{g_a\}$
\ElsIf{$i=i_b$}
\State $\bx'_{i_b}\leftarrow \bx_{i_b} \cup \{g_{b+1}\}$
\ElsIf{$i=i_c$ with $b<c<a$}
\State $\bx'_{i_c}\leftarrow \bx_{i_c} \cup \{g_{c+1}\} \setminus \{g_c\}$
\Else
\State $\bx'_i\leftarrow \bx_i$
\EndIf
\EndFor
\Return $\bx'$
\end{algorithmic}
\end{algorithm}

\section{Analysis of Our Algorithms}\label{sec:ana}

In this section, we analyze our algorithms and prove Theorem~\ref{thm:main}.
We show the correctness of the algorithm in Section~\ref{sec:cor}, and show its termination and polynomial-time solvability when the number of agents is fixed in Section~\ref{sec:ter}. 

\subsection{Correctness}\label{sec:cor}
In this section, we show the correctness of our main algorithm (Algorithm~\ref{alg:1}) by proving the following theorem.
\begin{theorem}\label{thm:correctness}(Correctness of Algorithm~\ref{alg:1})
    Let $\bx$ be an output of Algorithm~\ref{alg:1}.
    Then, $\bx$ is an EF1 and fPO allocation in a given fair division instance $(N, M, V)$.
\end{theorem}
To prove Theorem~\ref{thm:correctness}, we first show that the {\it FindSolution} algorithm maintains the invariant.
Fix an arbitrary $k \in [n]$ and consider the point in Algorithm~\ref{alg:1}  where the {\it FindSolution} algorithm is called for the $k$-th time.
At this point, agent $k$ is the most recently added agent to $\mathcal{N}$.
We assume the invariant that $(\bx, \bp)$ is pEF1 for all agents except $k$ and that $\bx_i \subseteq {\rm MBB}_i$ for each $i\in \mathcal{N}$ for the input of {\it FindSolution} algorithm.

{\it FindSolution} algorithm consists of two types of steps: a {\it Transfer step}, which occurs when the {\it Transfer} algorithm is called, and a {\it Price-increasing step}, which occurs when prices are updated.

For our analysis, we introduce the concept of a {\it time step}. For any positive integer $t$, {\it time-step $t$} refers to the state of the algorithm {\it immediately before} either a transfer step or a price-increasing step occurs during the $t$-th iteration of the {\it FindSolution} algorithm (called during the $k$-th iteration of Algorithm~\ref{alg:1}).
Let $\bx^t=(\bx^t_i)_{i \in \mathcal{N}}$, $\bp^t$, and ${{\rm MBB}}^t_i$ denote the allocation, the price vector, and the MBB ratio for each agent $i$ , respectively, at time-step $t$.
\begin{lemma}\label{lem:mbb}
    At any time-step $t$, $\bx^t_i \subseteq {{\rm MBB}}^t_i$ for any $i\in \mathcal{N}$.
\end{lemma}
\begin{proof}    
    By our assumption, the condition $\bx^1_i \subseteq {{\rm MBB}}^1_i$ holds for any $i\in \mathcal{N}$.
    It remains to show that if $\bx^t_i \subseteq {{\rm MBB}}^t_i$ holds for any $i\in \mathcal{N}$, then $\bx^{t+1}_i \subseteq {{\rm MBB}}^{t+1}_i$ also holds for any $i\in \mathcal{N}$.
    
    If Algorithm~\ref{alg:2} executes a transfer step in the $t$-th iteration, reallocation occurs along a directed path in $G(\bx^t, \bp^t)$. 
    Thus, the condition $\bx^{t+1}_i \subseteq {{\rm MBB}}^{t+1}_i$ holds for any $i\in \mathcal{N}$.
    
    Now, suppose that Algorithm~\ref{alg:2} executes a price-increasing step in the $t$-th iteration.
    By Lemma~\ref{lem:price}, the price of all goods in $R_{\mathcal{M}}$ increase.
    For any agent $i\in \mathcal{N}\setminus R_{\mathcal{N}}$, $\bx^t_i\cap R_{\mathcal{M}}=\emptyset$ holds.
    Thus, increasing the prices of goods in $R_{\mathcal{M}}$ does not affect the bang-per-buck ratio of the goods in $\bx^t_i$ for 
    each $ i \in \mathcal{N}\setminus R_{\mathcal{N}}$, preserving $\bx^{t+1}_i \subseteq {{\rm MBB}}^{t+1}_i$.
    
    For any agent $i \in R_{\mathcal{N}}$, increasing the prices of goods in $R_{\mathcal{M}}$ lowers their MBB ratio.
    However, by the choice of $\beta_1$, the price increase stops as soon as a new MBB edge appears between $R_{\mathcal{N}}$ and $\mathcal{M} \setminus R_{\mathcal{M}}$. 
    This ensures that the new MBB ratio for any agent $i\in R_{\mathcal{N}}$ does not fall below its second-highest bang-per-buck ratio prior to the price increase.
    Additionally, since $\bx^t_i \subseteq {\rm MBB}^t_i$, the bang-per-buck ratios of all goods in $\bx^t_i$ are equal. Since the price of goods in $\bx^t_i$ is uniformly increased by $\beta$ during the price-increasing step, the bang-per-buck ratios in $\bx^{t+1}_i$ remain equal after this step.
    Thus, $\bx^{t+1}_i \subseteq {{\rm MBB}}^{t+1}_i$ holds for $i \in R_{\mathcal{N}}$.
    Therefore, $\bx^{t+1}_i \subseteq {{\rm MBB}}^{t+1}_i$ holds for any $i\in \mathcal{N}$.
\end{proof}

\begin{lemma}\label{lem:pEF1invariant}
    At any time-step $t$, $(\bx^t, \bp^t)$ is pEF1 except for agent $k$, i.e., $\bp^t(\bx^t_i)\ge \max(\bx^t, \hat{\bp}^t)$ for any $i \in \mathcal{N}\setminus \{k\}$.
\end{lemma}
\begin{proof}
By our assumption, $(\bx^1, \bp^1)$ is pEF1 except for agent $k$.
It remains to show that if $(\bx^t, \bp^t)$ is pEF1 except for agent $k$, then $(\bx^{t+1}, \bp^{t+1})$ is also pEF1 except for agent $k$, i.e., 
$${\rm For~any}~ i \in \mathcal{N}\setminus \{k\},~ \bp^t(\bx^t_i)\ge \max(\bx^t, \hat{\bp}^t) \implies \bp^{t+1}(\bx^{t+1}_i)\ge \max(\bx^{t+1}, \hat{\bp}^{t+1}).$$

First consider the case where Algorithm~\ref{alg:2} executes a price-increasing step in the $t$-th iteration.
By the choice of $\beta_2$, the price of goods in $R_{\mathcal{M}}$ increases only until an agent in $R_{\mathcal{N}}$ becomes a maximum violator.
Furthermore, since for any maximum violator $j\in K$, $\bx_j\cap R_{\mathcal{M}}=\emptyset$ holds, the total price of the goods held by agent $j$ remains unchanged.
These imply that $\max(\bx^{t}, \hat{\bp}^{t})=\max(\bx^{t+1}, \hat{\bp}^{t+1})$ holds.
In addition, since the prices of the goods are non-decreasing by Lemma~\ref{lem:price}, we obtain $\bp^{t+1}(\bx^{t+1}_i) \ge \bp^{t}(\bx^{t}_i) \ge \max(\bx^{t}, \hat{\bp}^{t})=\max(\bx^{t+1}, \hat{\bp}^{t+1})$ for any $i \in \mathcal{N}\setminus \{k\}$.

Now, suppose that Algorithm~\ref{alg:2} executes a transfer step in the $t$-th iteration.
We show the following two conditions.
\begin{enumerate}
    \item $\bp^{t+1}(\bx^{t+1}_i)\ge \max(\bx^t, \hat{\bp}^t)$ for any $i \in \mathcal{N}\setminus \{k\}$.
    \item $\max(\bx^t, \hat{\bp}^t)\ge \max(\bx^{t+1}, \hat{\bp}^{t+1})$.
\end{enumerate}
Note that the prices of all goods remain unchanged in the transfer step, i.e., $p^{t+1}_g=p^t_g$ holds for any $g \in \mathcal{M}$.

We first show that $\bp^{t+1}(\bx^{t+1}_i)\ge \max(\bx^t, \hat{\bp}^t)$ for any $i \in \mathcal{N}\setminus \{k\}$ by considering following four cases:
    \begin{itemize}
        \item $\bx^{t+1}_i=\bx^{t}_i$.\\
        We directly obtain
        $\bp^{t+1}(\bx^{t+1}_i)=\bp^{t}(\bx^{t}_i)\ge \max(\bx^t, \hat{\bp}^t)$, where the inequality follows from the assumption.
        \item $i=i_a$.\\
        Agent $i=i_a$ relinquishes $g_a$ in the transfer step.
        By the definition of index $a$, we have $\bp^{t}(\bx^{t}_{i_a}\setminus \{g_a\})\ge \max(\bx^t, \hat{\bp}^t)$.
        Thus, we obtain
        $\bp^{t+1}(\bx^{t+1}_{i_a})=\bp^{t}(\bx^{t}_{i_a}\setminus \{g_a\})\ge \max(\bx^t, \hat{\bp}^t)$. 
        \item $i=i_b$.\\
        Agent $i=i_b$ acquires $g_{b+1}$ in the transfer step.
        Thus, we have 
        $\bp^{t+1}(\bx^{t+1}_{i_b})=\bp^{t}(\bx^{t}_{i_b}\cup \{g_{b+1}\})> \bp^{t}(\bx^{t}_{i_b})\ge \max(\bx^t, \hat{\bp}^t)$, where the first inequality follows from Lemma~\ref{lem:price} and the second inequality follows from the assumption.
        \item Otherwise.\\
        Agent $i$ must be some agent between $i_a$ and $i_b$ on the directed path.
        Let $i=i_c$ for some $c$ with $b<c<a$.
        Agent $i=i_c$ relinquishes $g_c$ and acquires $g_{c+1}$.
        By the choice of index $b$, $\max(\bx^t, \hat{\bp}^t) < \bp(\bx^t_{i_c}\cup \{g_{c+1}\} \setminus \{g_c\})$ holds, otherwise $c$ should have been chosen as $b$.
        Thus, we obtain
        $\bp^{t+1}(\bx^{t+1}_{i_c})=\bp^{t}(\bx^{t}_{i_c}\cup \{g_{c+1}\}\setminus \{g_c\})> \max(\bx^t, \hat{\bp}^t)$.
    \end{itemize}
    
We next show that $\max(\bx^t, \hat{\bp}^t)\ge \max(\bx^{t+1}, \hat{\bp}^{t+1})$.
To show this, it suffices to show that $\max(\bx^t, \hat{\bp}^t)\ge \hat{\bp}^{t+1}(\bx^{t+1}_i)$ for any $i\in \mathcal{N}$.
Fix an arbitrary $i\in \mathcal{N}$.

First consider the case $i=k$.
Since $(\bx^t, \bp^t)$ is pEF1 except for agent $k$ and $(\bx^t, \bp^t)$ is not pEF1, we have $\max(\bx^t, \hat{\bp}^t)> \bp^{t}(\bx^{t}_k)$ by Observation~\ref{ob:pEF1}.
In the {\it Transfer algorithm}, agent $k$ either receives one good or none.
In either case, $\bp^{t}(\bx^{t}_k) \ge \hat{\bp}^{t+1}(\bx^{t+1}_k)$ holds.
Therefore, $\max(\bx^t, \hat{\bp}^t)> \bp^{t}(\bx^{t}_k) \ge \hat{\bp}^{t+1}(\bx^{t+1}_k)$ holds.

Now, suppose that $i\neq k$.
    \begin{itemize}
        \item $\bx^{t+1}_i=\bx^{t}_i$.\\
        We directly obtain
        $\max(\bx^t, \hat{\bp}^t)\ge \hat{\bp}^{t}(\bx^{t}_i)\ge \hat{\bp}^{t+1}(\bx^{t+1}_i)$, where the first inequality follows from the definition of $\max(\bx^t, \hat{\bp}^t)$.
        \item $i=i_a$.\\
        Agent $i=i_a$ relinquishes $g_a$ in the transfer step.
        Thus, we obtain
        $\max(\bx^t, \hat{\bp}^t)\ge \hat{\bp}^{t}(\bx^{t}_{i_a})\ge \hat{\bp}^{t}(\bx^{t}_{i_a}\setminus \{g_a\})=\hat{\bp}^{t+1}(\bx^{t+1}_{i_a})$, where the first inequality follows from the definition of $\max(\bx^t, \hat{\bp}^t)$ and the second inequality follows from the monotonicity of $\hat{\bp}^t$.
        \item $i=i_b$.\\
        Agent $i=i_b$ acquires $g_{b+1}$ in the transfer step.
        By the definition of index $b$, we have $\max(\bx^t, \hat{\bp}^t)\ge \bp^{t}(\bx^{t}_{i_b}\cup \{g_{b+1}\}\setminus \{g_b\})$.
        Thus, we obtain
        $\max(\bx^t, \hat{\bp}^t)\ge \bp^{t}(\bx^{t}_{i_b}\cup \{g_{b+1}\}\setminus \{g_b\}) \ge \hat{\bp}^{t}(\bx^{t}_{i_b}\cup \{g_{b+1}\})=\hat{\bp}^{t+1}(\bx^{t+1}_{i_b})$, where the second inequality follows from the definition of $\hat{\bp}^t$.
        \item Otherwise.\\
        Agent $i$ must be some agent between $i_a$ and $i_b$ on the directed path.
        Let $i=i_c$ for some $c$ with $b<c<a$.
        Agent $i=i_c$ relinquishes $g_c$ and acquires $g_{c+1}$.
        By the choice of index $a$,  $\max(\bx^t, \hat{\bp}^t)> \bp^{t}(\bx^{t}_{i_c}\setminus \{g_c\})$ holds, otherwise $c$ should have been chosen as $a$.
        Thus, we obtain
        $\max(\bx^t, \hat{\bp}^t)> \bp^{t}(\bx^{t}_{i_c}\setminus \{g_c\})\ge \hat{\bp}^{t}(\bx^{t}_{i_c}\cup \{g_{c+1}\}\setminus \{g_c\})=\hat{\bp}^{t+1}(\bx^{t+1}_{i_c})$, where the second inequality follows from the definition of $\hat{\bp}^t$.
    \end{itemize}

Therefore, we have $\bp^t(\bx^t_i)\ge \max(\bx^t, \hat{\bp}^t)$ for any $i \in \mathcal{N}\setminus \{k\}$ at any time-step $t$.
\end{proof}
By Lemma~\ref{lem:pEF1invariant}, at the beginning of each iteration of the \textit{FindSolution} algorithm, $(\bx, \bp)$ is pEF1 except for agent $k$. In addition, by the conditions of the iteration, $(\bx, \bp)$ is not pEF1. Therefore, by Observation~\ref{ob:pEF1}, $\bp(\bx_i) \ge \max(\bx, \hat{\bp}) > \bp(\bx_k)$ holds for any $i \in \mathcal{N} \setminus \{k\}$. Consequently, the minimum spender is only agent $k$.

\begin{corollary}\label{cor:L={k}}
At the beginning of each iteration of Algorithm~\ref{alg:2}, the minimum spender is only agent $k$, i.e., $L=\{k\}$ holds.
\end{corollary}
We are now ready to prove Theorem~\ref{thm:correctness}.

\begin{proof}[Proof of Theorem~\ref{thm:correctness}]
We begin by proving that after each iteration of Algorithm~\ref{alg:1}, the following conditions hold: (1) $(\bx, \bp)$ is pEF1, and (2) $\bx_i \subseteq {\rm MBB}_i $ for each $i \in \mathcal{N}$.

In the first iteration of Algorithm~\ref{alg:1}, the price of each good $g\in M_1$ is set as $p_g = \frac{v_{1g}}{m \cdot \max_{h \in M} v_{1h}}$, hence the bang-per-buck ratio for all goods in $M_1$ is equal to $m \cdot \max_{h \in M} v_{1h}$. Thus, we have $\bx_1 = M_1 \subseteq {\rm MBB}_1$. At this point, $(\bx, \bp)$ is clearly pEF1, thus the conditions are satisfied at the end of the first iteration.

Next, we assume that the conditions hold after the $k$-th iteration for any $k \in [n-1]$ and show that they also hold after the $k+1$-th iteration. In the $k+1$-th iteration of Algorithm~\ref{alg:1}, the initial allocation for agent $k+1$ is $M_{k+1}$, and the initial price $p_g$ of each good $g\in M_{k+1}$ is set to
\begin{equation*}
  p_g = \frac{v_{{k+1}g}\cdot \min_{h\in \mathcal{M}} p_h}{m\cdot \max_{h\in M} v_{{k+1}h}}.
\end{equation*}
For any $g \in M_{k+1}$ and any $h \in \mathcal{M}$, we have
\begin{equation*}
    \frac{v_{k+1g}}{p_g} = \frac{m\cdot \max_{h\in M} v_{{k+1}h}}{\min_{h \in \mathcal{M}} p_h} \geq \frac{v_{k+1h}}{p_h},
\end{equation*}
which implies that the bang-per-buck ratio of any good in $M_{k+1}$ is equal and at least as large as that of $h$. Hence, $\bx_{k+1}=M_{k+1} \subseteq {\rm MBB}_{k+1}$.

For any agent $j \in [k]$, by the operation of Algorithm~\ref{alg:1}, we have $ v_{jg} = 0 $ for all $g \in M_{k+1}$, hence the MBB ratio for agent $j$ is unaffected by the addition of $M_{k+1}$. Therefore, just before calling Algorithm~\ref{alg:2} for the $k+1$-th time, $ \bx_i \subseteq {\rm MBB}_i$ holds for each $ i \in [k+1]$.

At the beginning of the $k+1$-th iteration of Algorithm~\ref{alg:1}, 
for any agent $j \in [k]$, agent $j$ holds at least one good.
Indeed, suppose, for contradiction, that there exists an agent in $[k]$ holding no goods. By assumption, at the beginning of the $k+1$-th iteration, $(\bx, \bp)$ is pEF1, and by Lemma~\ref{lem:price}, the price of each good is positive. This implies that each agent holds at most one good. Thus, $|\mathcal{N}| > |\mathcal{M}|$, which contradicts Observation~\ref{ob:hall} stating that $(\mathcal{N}, \mathcal{M}, \mathcal{V})$ satisfies Hall's condition.

Thus, since $\bp(\bx_{k+1}) \leq m \cdot \max_{g \in M_{k+1}} p_g \leq \min_{h \in \mathcal{M}} p_h \le \bp(\bx_j)$ holds for any agent $j\in [k]$, no agent $j \in [k]$ price envy agent $k+1$. Therefore, just before calling Algorithm~\ref{alg:2} for the $k+1$-th time, $(\bx, \bp)$ is pEF1 except for agent $k+1$.

By Lemmas~\ref{lem:mbb} and~\ref{lem:pEF1invariant}, the invariants are maintained during Algorithm~\ref{alg:2}. Therefore, once Algorithm~\ref{alg:2} terminates, both conditions (1) $(\bx, \bp)$ is pEF1, and (2) $ \bx_i \subseteq {\rm MBB}_i $ for each $ i \in \mathcal{N} $, are satisfied.

Finally, since we assume that for each good $g \in M$, there exists an agent $i \in N$ such that $v_{ig}>0$, $(\mathcal{N}, \mathcal{M}, \mathcal{V})=(N,M,V)$ holds after all iterations.
Therefore, by Lemma~\ref{lem:pEF1}, we conclude that the output of Algorithm~\ref{alg:1} is an EF1 and fPO allocation for the given fair division instance $(N,M,V)$.
\end{proof}

\subsection{Termination}\label{sec:ter}
In this section, we show the termination of our algorithms and its polynomial-time complexity when the number of agents is fixed by proving the following theorem. 
\begin{theorem}(Running time bound for Algorithm~\ref{alg:1})\label{thm:termination}
    Algorithm~\ref{alg:1} always terminates. Furthermore, if the number of agents is fixed, the algorithm runs in polynomial time.
\end{theorem}
To prove Theorem~\ref{thm:termination}, we first establish an upper bound on the number of iterations of the {\it FindSolution} algorithm.
Fix an arbitrary $k \in [n]$ and consider the point in Algorithm~\ref{alg:1} where Algorithm~\ref{alg:2} is called for the $k$-th time.
At this point, agent $k$ is the most recently added agent to $\mathcal{N}$.

For each $i\in \mathcal{N}$, we define the {\it level} of agent $i$ at time-step $t$, denoted by $\mathrm{level}(i,t)$, as half the length of the shortest directed path from agent $k$ to agent $i$ in the augmented MBB graph $G(\bx^t, \bp^t)$.
If no such path exists, we define the level of $i$ to be $k$.
Note that $\mathrm{level}(i,t) \in \{0,1,\ldots, k\}$.

For $r\in \{0\}\cup [k]$ and time-step $t$, we define $\mathcal{M}^t_r$ as the set of goods held by agents of level $r$ at time-step $t$, and $K^t$ as the set of maximum violators at time-step $t$.
The potential vector $\Phi_t$ is defined as follows:
\begin{equation*}
\Phi_t :=(|\mathcal{M}^t_0|,\ldots , |\mathcal{M}^t_k|, |K^t|).    
\end{equation*}

We consider the lexicographic ordering of the potential vector.
For two integer vectors $x=(x_1,\ldots , x_r)$ and $y=(y_1,\ldots , y_r)$, we say that $x$ is {\it lexicographically smaller} than $y$, denoted by $x \prec_{lex} y$ if
there exists an index $i\in [r]$ such that $x_i < y_i$ and $x_j=y_j$ for all $j < i$.
We denote $R_{\mathcal{N}}$ and $R_{\mathcal{M}}$ at time-step $t$ as $R^t_{\mathcal{N}}$ and $R^t_{\mathcal{M}}$.
\begin{lemma}\label{lem:find}
    The FindSolution algortihm terminates after at most $(k-1) \left( \frac{m+k}{k} \cdot \mathrm{e} \right)^{k}$ iterations.
\end{lemma}
\begin{proof}
Let $t$ be any time-step.
We show that in each iteration of Algorithm~\ref{alg:2}, the potential vector $\Phi_t$ either increases lexicographically, i.e., $\Phi_t \prec_{lex} \Phi_{t+1}$, or the iteration terminates. If this is true, the number of iterations of Algorithm~\ref{alg:2} is bounded by the number of distinct states of the potential vector $\Phi_t$. The values $|\mathcal{M}^t_0|,\ldots,|\mathcal{M}^t_k|$, and $|K^t|$ are all non-negative integers, and we have $|\mathcal{M}^t_0|+ \cdots +|\mathcal{M}^t_k| = |\mathcal{M}| \leq m$, and since agent $k$ is not a maximum violator in any iteration, $1\le |K^t| \leq k-1$ holds. Therefore, the number of iterations of Algorithm~\ref{alg:2} is at most $(k-1) \binom{m+k}{k} \leq (k-1) \left( \frac{m+k}{k} \cdot \mathrm{e} \right)^{k}$, where we use $\binom{p}{q} \le \left(\frac{p}{q}\cdot \mathrm{e}\right)^q$.

Now, let us show that $\Phi_t \prec_{lex} \Phi_{t+1}$, or the iteration terminates in the $t$-th iteration. First, consider the case where Algorithm~\ref{alg:2} executes a price-increasing step in the $t$-th iteration. 
We first show that in a price-increasing step, the levels of agents below level $k-1$ do not change, i.e., $\mathrm{level}(i,t)\le k-1 \implies \mathrm{level}(i,t+1) = \mathrm{level}(i,t)$.
Indeed, since Algorithm~\ref{alg:2} uniformly increases the prices of goods in $R^t_{\mathcal{M}}$, and by the choice of $\beta_1$, the status of the MBB edges between $R^t_{\mathcal{N}}$ and $R^t_{\mathcal{M}}$ remains unchanged before and after the price-increasing step. 
Additionally, the status of the allocation edges also does not change before and after the price-increasing step.
Thus, if $\mathrm{level}(i,t)\le k-1$, then $\mathrm{level}(i,t+1) \le \mathrm{level}(i,t)$ holds. Moreover, if $\mathrm{level}(i,t+1) < \mathrm{level}(i,t)$ holds, then the shortest directed path $P$ from agent $k$ to $i$ in $G(\bx^{t+1}, \bp^{t+1})$ must use at least one newly created edge between $R^t_{\mathcal{N}}$ and $\mathcal{M} \setminus R^t_{\mathcal{M}}$.
This implies that $P$ must use an MBB edge from $\mathcal{N}\setminus R^t_{\mathcal{N}}$ to $R^t_{\mathcal{M}}$.
However, since the price increase rate $\beta$ is greater than $1$, 
there are no MBB edges between $\mathcal{N}\setminus R^t_{\mathcal{N}}$ and $R^t_{\mathcal{M}}$ at time-step $t+1$, which is a contradiction.

Next, we analyze the cases based on which event determined the price increase rate $\beta$.

If $\beta = \beta_1$, there exists some $j \in R^t_{\mathcal{N}}$ and some $g \in \mathcal{M} \setminus R^t_{\mathcal{M}}$ such that $(j, g)$ forms an MBB edge at time-step $t+1$. In this case, the level of the agent who held good $g$ at time-step $t$ is reduced from $k$ to at most $k-1$. Thus, since the levels of agents below level $k-1$ remain unchanged and at least one agent with level $k$ who holds at least one good drops to a level below $k$, we have $\Phi_t \prec_{lex} \Phi_{t+1}$.

If $\beta = \beta_2$, at least one new maximum violator is added after price-increasing step, and maximum violators at time-step $t$ remain as maximum violators at time-step $t+1$, implying that $|K^t| < |K^{t+1}|$ holds. Therefore, since the levels of agents below level $k-1$ remain unchanged and the number of maximum violators increases, it follows that $\Phi_t \prec_{lex} \Phi_{t+1}$.

If $\beta = \beta_3$, since the solution $(\bx^{t+1}, \bp^{t+1})$ becomes pEF1, the iteration terminates.

Now, suppose that Algorithm~\ref{alg:2} executes a transfer step in the $t$-th iteration. Let $P = (i_0, g_1, i_1, g_2, \ldots, i_{\ell-1}, g_\ell, i_\ell)$ be the shortest directed path computed by Algorithm~\ref{alg:3}, where $i_0 = k$ and $i_\ell \in K^t$. Let the indices $a \in [l]$ and $b \in \{0\}\cup [a-1]$ be those computed by Algorithm~\ref{alg:3}.
We observe that for any $0 \leq r \leq \ell$, $\mathrm{level}(i_r, t)=r$ holds. Indeed, it clearly holds that $\mathrm{level}(i_r, t)\le r$ as the existence of $P$. If $\mathrm{level}(i_r, t)<r$, there exists a directed path from $k$ to $i_r$ of the length is less than $r$, contradicting the minimality of the length of $P$.

We next establish the following three claims.
\begin{enumerate}
    \item After the transfer step, the bundle size of $i_b$ increases by exactly 1, i.e., $|\bx^{t+1}_{i_b}| = |\bx^{t}_{i_b}| + 1$.
    \item After the transfer step, the bundle size of agents other than $i_b$ with level at most $b$ do not change, i.e., $\mathrm{level}(i, t)\le b {\rm ~and~} i\neq i_b \implies |\bx^{t+1}_{i}| = |\bx^{t}_{i}|$
    \item After the transfer step, the levels of agents with level at most $b$ remain unchanged, i.e., $\mathrm{level}(i, t) \le b \implies \mathrm{level}(i, t+1)=\mathrm{level}(i, t)$.
\end{enumerate}

It is clear from the operation of Algorithm~\ref{alg:3} that the bundle size of $i_b$ increases by exactly 1.

To see the second part, the only agents whose bundle sizes change are $i_a$ and $i_b$. Since the level of $i_a$ is $a > b$, the claim holds.

Now, let us show the third part.
Let $i$ be an agent such that $\mathrm{level}(i, t) \le b$, and let $Q$ denote the shortest directed path from agent $k$ to agent $i$ in $G(\bx^t, \bp^t)$.
Since the transfer step at time-step $t$ does not affect $Q$, $Q$ still exists in $G(\bx^{t+1}, \bp^{t+1})$. 
This implies that $\mathrm{level}(i, t+1) \le \mathrm{level}(i, t)$.
We derive a contradiction by assuming $\mathrm{level}(i, t+1) < \mathrm{level}(i, t)$.

Let $Q'$ denote the shortest directed path from agent $k$ to agent $i$ in $G(\bx^{t+1}, \bp^{t+1})$. 
Since a transfer step does not change the MBB edges and the assumption of $\mathrm{level}(i, t+1) < \mathrm{level}(i, t)$, $Q'$ must include a newly created allocation edge (which did not exist at time-step $t$ but exists at time-step $t+1$). Let $(g, j)$ be such an allocation edge where the level of $j$ is minimal among such edges. Let $j'$ be the agent on $Q'$ whose level is exactly one less than $j$.
By the choice of $j$, the directed path from $k$ to $j'$ on $Q'$ exists at both time-step $t$ and $t+1$.
Since $(g, j)$ is a newly created allocation edge, some agent other than agent $j$ must have good $g$ at time-step $t$. Let $j''$ denote the agent who held $g$ at time-step $t$.

If agent $j''$ matches an agent on the path $Q'$ from agent $k$ to $j'$, let $Q''$ denote the subpath of $Q'$ from agent $k$ to agent $j''$.
We have
\begin{equation*}
   b < \text{level}(j'', t) \le \text{level}(j'', t+1) < \text{level}(j, t+1) \le \text{level}(i, t+1) < \text{level}(i, t) \le b, 
\end{equation*}
where the first inequality follows from the fact that agent $j''$ relinquishes good $g$ in the transfer step at time-step $t$, the second inequality follows from the existence of $Q''$ at both time-step $t$ and $t+1$, and the minimality of the length of $Q'$, 
leading to a contradiction.

Otherwise, by combining the directed path from $k$ to $j'$ on $Q'$ with the path $(j', g, j'')$, we obtain a directed path $Q'''$ from agent $k$ to agent $j''$ at time-step $t$. Therefore, we have
\begin{equation*}
   b < \text{level}(j'', t) \le \text{level}(j', t) + 1 \le \text{level}(j', t+1) + 1 \le \text{level}(j, t+1) \le \text{level}(i, t+1) < \text{level}(i, t) \le b, 
\end{equation*}
where the first inequality follows from the fact that agent $j''$ relinquishes good $g$ in the transfer step at time-step $t$, the second inequality follows from the construction of $Q'''$, and the third inequality follows from the the existence of $Q'$ at both time-step $t$ and $t+1$, and the minimality of the length of $Q'$, 
leading to a contradiction. Hence, the third claim holds.

From the three claims above, it follows that $\Phi_t \prec_{lex} \Phi_{t+1}$.
This completes the proof of Lemma~\ref{lem:find}.

\end{proof}

We are now ready to prove Theorem~\ref{thm:termination}.
\begin{proof}[Proof of Theorem~\ref{thm:termination}]
By Lemma~\ref{lem:find}, each iteration of Algorithm~\ref{alg:1} terminates. Thus, Algorithm~\ref{alg:1} is guaranteed to terminate.

We now show that if the number of agents is constant, the algorithm runs in polynomial time. If $n$ is constant, each $k \in [n]$ is also constant, and thus, by Lemma~\ref{lem:find}, the number of iterations of Algorithm~\ref{alg:2} is bounded by $\mathbf{poly}(m)$.

Algorithm~\ref{alg:1} performs an iterative process, sequentially adding each agent to the instance. This loop repeats $n$ times, executing a sequence of operations for each agent. The computation of $M_k$ and the initial price $p_g$ for each good can both clearly be completed in polynomial time.

Algorithm~\ref{alg:2} iterates while the current solution $(\bx, \bp)$ is not pEF1 in the instance $(\mathcal{N}, \mathcal{M}, \mathcal{V})$. Checking whether $(\bx, \bp)$ satisfies pEF1 can be performed in polynomial time. Additionally, determining whether $K \cap R_{\mathcal{N}}$ is empty set can be achieved by constructing $G(\bx, \bp)$ and using a breadth-first search to compute $R_{\mathcal{N}}$, which is also feasible in polynomial time. Furthermore, the calculations of $\beta_1$, $\beta_2$, and $\beta_3$, as well as the price updates, are straightforward polynomial-time operations.

Algorithm~\ref{alg:3} performs the search for the shortest path and the transfer of goods between agents. Constructing $G(\bx, \bp)$ and finding the shortest path are both achievable in polynomial time. Similarly, the computation of indices $a$ and $b$, along with the item transfer operations, are clearly executable within polynomial time.

From the above analysis, we conclude that  Algorithm~\ref{alg:1} runs in polynomial time when the number of agents is fixed.

\end{proof}

Theorem~\ref{thm:main} follows directly from Theorem~\ref{thm:correctness} and~\ref{thm:termination}.
\section{Conclusion}\label{sec:con}
We study the problem of fairly and efficiently allocating indivisible goods among agents with additive valuation functions. We develop an algorithm that achieves both EF1 and fPO under additive valuations and show that it runs in polynomial time when the number of agents is fixed. Our approach incrementally adds agents to the instance, finding an allocation that satisfies EF1 and fPO at each step.

Our main algorithm (Algorithm~\ref{alg:1}) runs in polynomial time if the number of iterations of Algorithm~\ref{alg:2} is polynomially bounded. An interesting open question is whether there exist instances where this approach fails to run in polynomial time. Additionally, refining the choice of the shortest directed path $P$ in Algorithm~\ref{alg:3} may yield further improvements in time complexity.

Determining whether a polynomial-time algorithm exists for finding EF1 and PO (or fPO) allocations remains an important direction for future research. Another intriguing question is the existence of EF1 and PO (or fPO) allocations under more general valuation classes beyond additive valuations, as well as in the context of chore division.

\section*{Acknowledgments}
This work was partially supported by the joint project of Kyoto University and Toyota Motor Corporation, titled ``Advanced Mathematical Science for Mobility Society'' and supported by JST ERATO Grant Number JPMJER2301, and JSPS KAKENHI Grant Number JP23K19956, Japan.

\appendix
\section{Proof of Theorem~\ref{thm:nash}}\label{ap:1}
In this section, we prove that Theorem~\ref{thm:nash} follows from Theorem~\ref{thm:main}. Note that the proof method closely follows that of Theorem 3.3 in~\cite{barman2018finding}. To establish this result, we will employ the following lemma from~\cite{barman2018finding}.
\begin{lemma}(Lemma 3.4 in \cite{barman2018finding})\label{lem:id}
Given a fair division instance with identical and additive valuations, any EF1 allocation provides a $\mathrm{e}^{1/\mathrm{e}}$-approximation for the Nash social welfare maximization problem.
\end{lemma}

\begin{proof}[Proof of Theorem~\ref{thm:nash}]
For a given fair division instance $\mathcal{I} = (N, M, V)$, let $\bx$ and $\bp$ denote the allocation and the final price vector, respectively, as returned by Algorithm~\ref{alg:1}. Let $\alpha_i$ be the maximum bang-per-buck ratio for each agent $i$ with respect to $\bp$. Construct a scaled instance $\mathcal{I}' = (N, M, V')$ by setting $v'_{ig} = \frac{v_{ig}}{\alpha_i}$ for all $i\in N$ and $g\in M$. Then, for any allocation $\by$ in $\mathcal{I}'$, the Nash social welfare (NSW) of $\by$ in $\mathcal{I}$ is $\frac{1}{\left( \prod_{i \in N} \alpha_i \right)^{1/n}}$ times the NSW of $\by$ in $\mathcal{I}'$. Therefore, to establish the desired approximation guarantee, it suffices to show that $\bx$ achieves an approximation factor of $\mathrm{e}^{1/\mathrm{e}}$ in the scaled instance $\mathcal{I}'$.

Let $\bx^*$ denote a maximum Nash social welfare allocation in the original instance $\mathcal{I}$. From the above discussion, it follows that $\bx^*$ is also a maximum Nash social welfare allocation in the scaled instance $\mathcal{I}'$. Furthermore, for each agent $i$, we have $v'_{ig} = p_g$ for all $g \in \mathrm{MBB}_i$, and $v'_{ig} < p_g$ for all $g \notin \mathrm{MBB}_i$. Consequently, for any agent $i$, since $\bx_i \subseteq \mathrm{MBB}_i$, we have $v'_i(\bx_i) = \bp(\bx_i)$, and $v'_i(\bx^*_i) \leq \bp(\bx^*_i)$. As a result, the NSW of $\bx$ and $\bx^*$ in $\mathcal{I}'$ satisfy the following equality and inequality:

\begin{align*} \left( \prod_{i=1}^n v'_i(\bx_i) \right)^{1/n} &= \left( \prod_{i=1}^n \bp(\bx_i) \right)^{1/n}, \\
\left( \prod_{i=1}^n v'_i(\bx^*_i) \right)^{1/n} &\leq \left( \prod_{i=1}^n \bp(\bx^*_i) \right)^{1/n}.
\end{align*}

Furthermore, we construct an instance $\mathcal{I}^{\mathrm{id}} = (N, M, V^{\mathrm{id}})$ with identical valuations by setting $v^{\mathrm{id}}_{ig} = p_g$ for all $i \in N$ and $g \in M$. The allocation $\bx$ is pEF1 with respect to $\bp$ in $\mathcal{I}$. Thus, $\bx$ satisfies EF1 in $\mathcal{I}^{\mathrm{id}}$. By Lemma~\ref{lem:id}, the following inequality holds.
\begin{equation*}
    \left( \prod_{i=1}^n \bp(\bx_i) \right)^{1/n} \ge 
    \mathrm{e}^{-1/\mathrm{e}} \max_{\by \in \mathcal{X}} \left( \prod_{i=1}^n \bp(\by_i) \right)^{1/n} \ge 
    \mathrm{e}^{-1/\mathrm{e}} \left( \prod_{i=1}^n \bp(\bx^*_i) \right)^{1/n}, 
\end{equation*}
where $\mathcal{X}$ represents the set of all allocations.
Combining the above equality and inequality, we obtain 
\begin{equation*}
    \left( \prod_{i=1}^n v'_i(\bx_i) \right)^{1/n} \ge
    \mathrm{e}^{-1/\mathrm{e}} \left( \prod_{i=1}^n v'_i(\bx^*_i) \right)^{1/n}, 
\end{equation*}
which provides a similar approximation guarantee for the original instance $\mathcal{I}$.
Finally, by Theorem~\ref{thm:main}, the allocation $\bx$ can be computed in polynomial time when the number of agent is fixed.
This completes the proof of Theorem~\ref{thm:nash}.
\end{proof}
\section{Error Analysis of the Algorithm in~\cite{garg2024computing}}\label{ap:2}
In this section, we discuss an error in~\cite{garg2024computing}. In this paper, the authors claim the existence of a pseudo-polynomial-time algorithm for finding an allocation that satisfies both EF1 and fPO by modifying algorithm~\cite{barman2018finding}. However, as we will explain below, there is an error in the proof, which prevents us from verifying its correctness. Specifically, this concerns Lemma 4 on page 11 in~\cite{garg2024computing}.

In the proof of Lemma 4, it is stated that ``Since there are at most $n$ agents outside $C_L$ at any given iteration, there can only be at most $n$ price-rise steps." However, since a transfer step may cause an agent to leave the set $C_L$, this claim does not hold.

As a result, with the current proof, the claim of Lemma 4 cannot be established, leaving the existence of a pseudo-polynomial-time algorithm for finding an allocation that satisfies EF1 and fPO unresolved.
\bibliographystyle{alphaurl}
\bibliography{main}
\end{document}